\documentclass[11pt]{article}

\usepackage[margin=1in]{geometry}
\usepackage{amsmath,amssymb,amsthm,mathtools}
\usepackage{booktabs}
\usepackage{enumitem}
\usepackage{array}
\usepackage{microtype}
\usepackage{tikz}
\usetikzlibrary{arrows.meta,positioning}
\usepackage[hidelinks]{hyperref}
\hypersetup{
  pdftitle={When Is Forgetting Provenance Lawful? Endpoint Sufficiency and Behavioral Quotients of Generative Systems},
  pdfauthor={David Carr},
  pdfsubject={Endpoint projections, trace sufficiency, bisimulation, quotientability, and provenance-sensitive generative systems}
}

\newtheorem{theorem}{Theorem}[section]
\newtheorem{proposition}[theorem]{Proposition}
\newtheorem{corollary}[theorem]{Corollary}

\theoremstyle{definition}
\newtheorem{definition}[theorem]{Definition}
\newtheorem{example}[theorem]{Example}
\theoremstyle{remark}
\newtheorem{remark}[theorem]{Remark}

\newcommand{\Tr}{\operatorname{Tr}}
\newcommand{\En}{\operatorname{En}}
\newcommand{\Post}{\operatorname{Post}}
\newcommand{\eu}{\equiv_{\mathsf U}}
\newcommand{\etr}{\equiv_{\mathrm{tr}}}

\newcommand{\step}[1]{\xrightarrow{#1}}
\newcommand{\steps}[1]{\xRightarrow{#1}}

\newcommand{\quot}{\mathord{/}}

\title{When Is Forgetting Provenance Lawful?\\
Endpoint Sufficiency and Behavioral Quotients of Generative Systems}
\author{David Carr}
\date{July 2026\\\small Working draft 0.2}

\begin{document}
\maketitle

\begin{abstract}
A provenance-decorated generative system may contain distinct occurrences with the same visible endpoint. The central abstraction question is therefore exact: when may provenance be forgotten without changing the lawful future? We study this question through the labeled transition system induced by lawful generation and an endpoint projection $\mathsf U$ from occurrences to visible objects.

Three observation levels are separated. Enabled sufficiency preserves immediately available rule labels; trace sufficiency preserves all finite lawful rule-label traces; quotient sufficiency preserves the branching transition structure modulo endpoint equivalence. The resulting hierarchy is strict. At the linear-time level, endpoint trace sufficiency is exactly inclusion of endpoint equivalence in finite-trace equivalence. At the branching level, the canonical endpoint quotient is representative-independent exactly when endpoint equivalence is a strong bisimulation equivalence.

Two canonical repairs follow. Intersecting endpoint equivalence with trace equivalence gives the greatest endpoint-respecting relation preserving finite traces. The greatest endpoint-respecting bisimulation gives the maximally coarse exact branching quotient and is universal among endpoint-respecting exact quotients. For finite systems it is computed by a terminating partition-refinement procedure initialized by endpoint classes. A self-contained application to provenance-decorated nested recursive-recombinant generation exhibits two occurrences with the same visible graph $I\to A\to B$ but different enabled futures; the refinement procedure separates them in its first round. The result replaces an all-or-nothing demand to retain provenance with an exact criterion for retaining only the distinctions that remain behaviorally operative.
\end{abstract}

\section{Introduction}

A generative object can be represented at several levels. One representation records only its visible endpoint. Another records the occurrence that produced that endpoint, including ancestry, merge events, dependency structure, or other provenance. The endpoint representation is smaller. It is also potentially false as a state description.

The issue arises whenever admissibility depends on lineage. Two occurrences may project to isomorphic visible objects while admitting different lawful next moves. A previous provenance-decorated nested recursive-recombinant generative (NRRG) framework supplied a concrete witness: two pointed typed graphs with the same endpoint shape had different recombinant ancestry and therefore different permissions for a boundary-touching rewrite \cite{carr2026nrrg}. That result established failure. The present paper asks for the exact success criterion:

\begin{quote}
When does an endpoint projection retain enough information to determine lawful future behavior, and what is the largest amount of provenance that can safely be forgotten?
\end{quote}

The question requires one correction to the informal slogan that motivated it. ``Endpoint sufficiency'' can mean at least three different things. It may require only the same immediately enabled rule labels. It may require the same finite trace language. Or it may require the same branching transition structure after quotienting by endpoint. These are not equivalent. In standard process semantics, trace equivalence is weaker than bisimulation because linear-time observations can ignore branching distinctions \cite{park1981,milner1989,vanglabbeek2001,rutten2000}.

This paper therefore separates the levels and proves the corresponding criteria. Its main contributions are:

\begin{enumerate}[label=(\arabic*)]
    \item a formulation of endpoint sufficiency relative to a fixed lawful-transition encoding, endpoint projection, label vocabulary, and observation level;
    \item a strict three-level hierarchy separating enabled, finite-trace, and exact branching sufficiency;
    \item an exact quotient criterion: the endpoint quotient is representative-independent if and only if endpoint equivalence is a strong bisimulation equivalence;
    \item two canonical repairs: the greatest endpoint-respecting trace-safe refinement and the greatest endpoint-respecting bisimulation;
    \item a universal maximal-forgetting property showing that every exact endpoint-respecting quotient refines the canonical branching quotient;
    \item a finite partition-refinement procedure computing that quotient from endpoint classes; and
    \item a self-contained NRRG witness showing in one diagram why endpoint-isomorphic occurrences can require an immediate behavioral split.
\end{enumerate}

The transition-system tools used below are standard: trace equivalence, bisimulation, coinduction, and partition refinement have mature theories. The contribution here is not a reinvention of those tools. It is their constrained deployment at the endpoint--provenance boundary of generative systems: the problem formulation, the separation of three non-equivalent sufficiency claims, and the derivation of a maximal lawful-forgetting quotient that is required to respect visible endpoints.

The paper does not claim that all provenance must be retained. Its conclusion is more discriminating. Provenance distinctions may be forgotten exactly to the extent that they are behaviorally inert relative to the chosen observation level.

\section{Lawful Generative Transition Systems}

The results apply to any system whose lawful one-step continuations can be represented as labeled transitions.

\begin{definition}[Lawful generative transition system]
A lawful generative transition system is a tuple
\[
\mathcal T=(S,\Sigma,\longrightarrow,\mathsf U,X,\cong_X),
\]
where:
\begin{enumerate}[label=(\arabic*)]
    \item $S$ is a set of provenance-decorated generated occurrences;
    \item $\Sigma$ is a set of rule labels;
    \item $\longrightarrow\;\subseteq S\times\Sigma\times S$ is the lawful labeled transition relation;
    \item $\mathsf U:S\to X$ is an endpoint projection into a set $X$ of visible objects; and
    \item $\cong_X$ is an equivalence relation representing visible endpoint isomorphism.
\end{enumerate}
We write $s\step{a}t$ when $(s,a,t)\in\longrightarrow$.
\end{definition}

The transition relation is already restricted to lawful generation. An admissibility law $\tau$ therefore appears indirectly, by determining which proposed generation events become edges.

\begin{remark}[Recursive and recombinant generation]
A recursive event naturally gives a unary transition. A recombinant event has several parents. To obtain an ordinary labeled transition system, choose a focal parent and treat the remaining parents as an environment: $s\step{C_i}t$ means that $s$ occurs in parent position $i$ of some lawful recombinant event producing $t$. The label may record only the rule schema, or may additionally record parent role, partner type, or partner endpoint. If partner identity itself must remain explicit, one may instead use a hypertransition system; that extension is left open.
\end{remark}

\begin{remark}[Sufficiency is relative to an encoding]
No sufficiency judgment below is absolute. It is relative to the chosen state space $S$, transition encoding $\longrightarrow$, label vocabulary $\Sigma$, endpoint projection $\mathsf U$, endpoint isomorphism $\cong_X$, and observation level. Refining labels or moving recombinant-partner data from the environment into the state can destroy a previously valid quotient; coarsening the observation may create one. The focal-parent convention is therefore a declared modeling choice, not a theorem about all multi-parent semantics. Once these choices are fixed, the criteria below are exact.
\end{remark}

\begin{definition}[Endpoint equivalence]
The endpoint equivalence induced by $\mathsf U$ is
\[
s\eu t
\quad\Longleftrightarrow\quad
\mathsf U(s)\cong_X\mathsf U(t).
\]
Its equivalence classes are the visible endpoint classes.
\end{definition}

\subsection{Paths, traces, and enabled labels}

For $w=a_1\cdots a_n\in\Sigma^*$, write $s\steps{w}t$ when there exist states
\[
s=s_0\step{a_1}s_1\step{a_2}\cdots\step{a_n}s_n=t.
\]
For the empty word, $s\steps{\varepsilon}s$.

\begin{definition}[Finite trace language and enabled set]
For $s\in S$, define
\[
\Tr(s)=\{w\in\Sigma^*: \exists t\in S,\ s\steps{w}t\}
\]
and
\[
\En(s)=\{a\in\Sigma: \exists t\in S,\ s\step{a}t\}.
\]
Two states are finite-trace equivalent, written $s\etr t$, when $\Tr(s)=\Tr(t)$.
\end{definition}

The trace language records which finite rule-label sequences can occur. It does not record where a branch occurs, how many distinct successors realize a trace, or whether two traces share an initial continuation.

\section{Three Levels of Endpoint Sufficiency}

\begin{definition}[Enabled sufficiency]
The endpoint projection $\mathsf U$ is \emph{enabled-sufficient} when
\[
s\eu t\quad\Longrightarrow\quad \En(s)=\En(t).
\]
\end{definition}

\begin{definition}[Trace sufficiency]
The endpoint projection $\mathsf U$ is \emph{trace-sufficient} when
\[
s\eu t\quad\Longrightarrow\quad \Tr(s)=\Tr(t).
\]
\end{definition}

To define the branching criterion, let $E$ be any equivalence relation on $S$. For $a\in\Sigma$, define the set of successor $E$-classes
\[
\Post_{E,a}(s)=\{[t]_E: s\step{a}t\}.
\]

\begin{definition}[Quotient sufficiency]
The endpoint projection $\mathsf U$ is \emph{quotient-sufficient} when, for every $s\eu t$ and every $a\in\Sigma$,
\[
\Post_{\eu,a}(s)=\Post_{\eu,a}(t).
\]
Equivalently, the successor endpoint classes available under each label depend only on the current endpoint class and not on the chosen occurrence representative.
\end{definition}

This is the strongest of the three notions. It retains the branching structure modulo visible endpoint equivalence.

\section{Trace Sufficiency and Its Canonical Repair}

\begin{proposition}[Trace criterion]
The endpoint projection $\mathsf U$ is trace-sufficient if and only if
\[
\eu\;\subseteq\;\etr.
\]
\end{proposition}

\begin{proof}
By definition, trace sufficiency asserts that whenever $s\eu t$, the trace languages of $s$ and $t$ are equal. Equality of trace languages is exactly $s\etr t$. Thus the assertion is equivalent to inclusion of relations.
\end{proof}

Although elementary, the proposition identifies the correct exact criterion at the linear-time level. It also yields a canonical repair.

\begin{definition}[Canonical trace-safe refinement]
Define
\[
E^{\mathrm{tr}}_{\mathsf U}=\eu\cap\etr.
\]
\end{definition}

\begin{proposition}[Coarsest endpoint-respecting trace-safe refinement]
$E^{\mathrm{tr}}_{\mathsf U}$ is the greatest equivalence relation $R$ satisfying both:
\begin{enumerate}[label=(\arabic*)]
    \item $R\subseteq\eu$; and
    \item $sRt$ implies $\Tr(s)=\Tr(t)$.
\end{enumerate}
\end{proposition}

\begin{proof}
Both $\eu$ and $\etr$ are equivalence relations, so their intersection is an equivalence relation. It plainly satisfies the two conditions. If $R$ also satisfies them, then every pair in $R$ belongs to both $\eu$ and $\etr$, hence $R\subseteq E^{\mathrm{tr}}_{\mathsf U}$.
\end{proof}

Thus $E^{\mathrm{tr}}_{\mathsf U}$ identifies exactly those endpoint-equivalent occurrences whose finite lawful futures are indistinguishable as rule-label traces.

\section{The Exact Quotient Criterion}

We recall the standard branching-time notion needed for an exact endpoint quotient \cite{park1981,milner1989,rutten2000}.

\begin{definition}[Strong bisimulation]
A relation $R\subseteq S\times S$ is a strong bisimulation when, for every $(s,t)\in R$ and every $a\in\Sigma$:
\begin{enumerate}[label=(\arabic*)]
    \item if $s\step{a}s'$, then there exists $t\step{a}t'$ with $(s',t')\in R$;
    \item if $t\step{a}t'$, then there exists $s\step{a}s'$ with $(s',t')\in R$.
\end{enumerate}
\end{definition}

\begin{theorem}[Exact quotient criterion]
For the endpoint equivalence $\eu$, the following are equivalent:
\begin{enumerate}[label=(\roman*)]
    \item $\mathsf U$ is quotient-sufficient;
    \item $\eu$ is a strong bisimulation equivalence;
    \item for every endpoint class $A\in S\quot\eu$ and label $a$, the set
    \[
    \{B\in S\quot\eu: s\step{a}t\text{ for some }t\in B\}
    \]
    is independent of the representative $s\in A$.
\end{enumerate}
\end{theorem}

\begin{proof}
Conditions (i) and (iii) are the same representative-independence statement written in state and quotient notation.

Assume (i). Let $s\eu t$ and $s\step{a}s'$. Then $[s']_{\eu}\in\Post_{\eu,a}(s)$. Quotient sufficiency gives
\[
\Post_{\eu,a}(s)=\Post_{\eu,a}(t),
\]
so there exists $t\step{a}t'$ with $[t']_{\eu}=[s']_{\eu}$, hence $s'\eu t'$. The reverse matching condition is symmetric. Therefore $\eu$ is a strong bisimulation.

Conversely, assume (ii). Let $s\eu t$. If an endpoint class $B$ belongs to $\Post_{\eu,a}(s)$, choose $s\step{a}s'$ with $[s']_{\eu}=B$. Bisimulation supplies $t\step{a}t'$ with $s'\eu t'$, so $[t']_{\eu}=B$ and $B\in\Post_{\eu,a}(t)$. Thus one successor-class set is included in the other; the reverse inclusion follows symmetrically.
\end{proof}

\begin{definition}[Canonical endpoint quotient]
When the equivalent conditions hold, define the endpoint quotient transition system
\[
\mathcal T\quot\eu=(S\quot\eu,\Sigma,\step{}_{\eu})
\]
by
\[
A\step{a}_{\eu}B
\quad\Longleftrightarrow\quad
\text{there exist }s\in A\text{ and }t\in B\text{ with }s\step{a}t.
\]
By quotient sufficiency, this is equivalent to requiring that every representative $s\in A$ have some $a$-successor in $B$. Thus the outgoing target classes are independent of the representative used.
\end{definition}

\begin{corollary}[Exact trace preservation by a quotient-sufficient endpoint]
If $\mathsf U$ is quotient-sufficient, then for every $s\in S$,
\[
\Tr_{\mathcal T}(s)=\Tr_{\mathcal T\quot\eu}([s]_{\eu}).
\]
In particular, $\mathsf U$ is trace-sufficient.
\end{corollary}

\begin{proof}
Every concrete path projects to a quotient path. Conversely, use bisimulation step matching inductively along a quotient path to lift each quotient transition from any representative of its source class. Therefore the same finite words are realizable in both systems. If $s\eu t$, both project to the same quotient state and hence have the same trace language.
\end{proof}

\begin{corollary}[Hierarchy]
For every lawful generative transition system,
\[
\text{quotient sufficiency}
\Longrightarrow
\text{trace sufficiency}
\Longrightarrow
\text{enabled sufficiency}.
\]
\end{corollary}

\begin{proof}
The first implication is the preceding corollary. For the second, $a\in\En(s)$ exactly when the one-letter word $a$ belongs to $\Tr(s)$.
\end{proof}

\section{The Implications Are Strict}

The informal sentence ``endpoint sufficiency holds exactly when endpoint equivalence preserves lawful transition'' is correct only when \emph{sufficiency} means quotient sufficiency. It is too strong for trace sufficiency.

\begin{example}[Enabled-sufficient but not trace-sufficient]
Let the only nontrivial endpoint class be $\{x,y\}$, and let all other states have distinct endpoints. Take transitions
\[
x\step{a}p,
\qquad
y\step{a}q,
\qquad
p\step{b}r,
\]
with no other transitions. Then
\[
\En(x)=\En(y)=\{a\},
\]
but $ab\in\Tr(x)$ while $ab\notin\Tr(y)$. Thus enabled sufficiency does not imply trace sufficiency.
\end{example}

\begin{example}[Trace-sufficient but not quotient-sufficient]
Again let the only nontrivial endpoint class be $\{x,y\}$. Take transitions
\[
x\step{a}p,
\qquad
y\step{a}q,
\]
where $p$ and $q$ are terminal states with distinct visible endpoints. Then
\[
\Tr(x)=\Tr(y)=\{\varepsilon,a\},
\]
so the endpoint projection is trace-sufficient on $\{x,y\}$. However,
\[
\Post_{\eu,a}(x)=\{[p]_{\eu}\}
\neq
\{[q]_{\eu}\}=\Post_{\eu,a}(y).
\]
Hence the endpoint quotient is not representative-independent. Trace sufficiency does not imply quotient sufficiency.
\end{example}

The second example is the minimal warning against conflating a linear-time future language with a branching state description. The two occurrences admit the same label words but move into visibly different endpoint classes.

\section{The Greatest Endpoint-Respecting Bisimulation}

When raw endpoint equivalence is not quotient-sufficient, the correct repair is not necessarily to retain every provenance detail. We seek the coarsest refinement of endpoint classes that supports an exact branching quotient.

\begin{definition}[Endpoint-respecting bisimulation]
A strong bisimulation $R$ is endpoint-respecting when $R\subseteq\eu$.
\end{definition}

Let $\mathcal B_{\mathsf U}$ be the family of all endpoint-respecting strong bisimulations, and define
\[
\approx_{\mathsf U}
\;=\;
\bigcup_{R\in\mathcal B_{\mathsf U}}R.
\]

\begin{theorem}[Canonical maximal safe forgetting]
The relation $\approx_{\mathsf U}$ has the following properties:
\begin{enumerate}[label=(\roman*)]
    \item it is a strong bisimulation;
    \item it is an equivalence relation;
    \item it is contained in endpoint equivalence;
    \item it is the greatest endpoint-respecting strong bisimulation; and
    \item the quotient $\mathcal T\quot\approx_{\mathsf U}$ is the maximally coarse exact branching quotient that preserves visible endpoint distinctions.
\end{enumerate}
\end{theorem}

\begin{proof}
If $(s,t)\in\approx_{\mathsf U}$, then $(s,t)$ belongs to some endpoint-respecting bisimulation $R$. Every transition from either state can be matched by a transition to an $R$-related state, hence to an $\approx_{\mathsf U}$-related state. Therefore the union is a bisimulation. It is contained in $\eu$ because every member relation is.

The identity relation is an endpoint-respecting bisimulation, so $\approx_{\mathsf U}$ is reflexive. If $R$ is an endpoint-respecting bisimulation, then so is $R^{-1}$; hence the union is symmetric. If $s\approx_{\mathsf U}t$ and $t\approx_{\mathsf U}u$, choose endpoint-respecting bisimulations $R,Q$ with $sRt$ and $tQu$. The relational composition $R\circ Q$ is a strong bisimulation, and it remains contained in $\eu$ because $\eu$ is transitive. Thus $s\approx_{\mathsf U}u$, proving transitivity.

Greatestness follows from the definition as a union of all endpoint-respecting bisimulations. Since it is a bisimulation equivalence, its quotient is exact by the preceding theorem. Any other exact quotient preserving endpoint distinctions has a kernel that is an endpoint-respecting bisimulation equivalence, and hence is contained in $\approx_{\mathsf U}$. Therefore no such quotient can identify more states.
\end{proof}

\begin{corollary}[Universal property of maximal lawful forgetting]
Let $E\subseteq\eu$ be any strong bisimulation equivalence, and let $\pi_E:S\to S\quot E$ and $\pi_{\mathsf U}:S\to S\quot\approx_{\mathsf U}$ be the quotient maps. Then there is a unique surjective labeled-transition morphism
\[
q_E:S\quot E\longrightarrow S\quot\approx_{\mathsf U}
\]
satisfying $q_E\circ\pi_E=\pi_{\mathsf U}$. Hence every exact endpoint-respecting quotient refines the canonical maximal-forgetting quotient.
\end{corollary}

\begin{proof}
Because $E$ is an endpoint-respecting bisimulation, greatestness gives $E\subseteq\approx_{\mathsf U}$. Define
\[
q_E([s]_E)=[s]_{\approx_{\mathsf U}}.
\]
The inclusion makes this map well-defined. It is surjective, preserves labeled transitions by construction of quotient transitions, and is the only map commuting with the two quotient projections.
\end{proof}

\begin{corollary}[Branching refinement is trace-safe]
\[
\approx_{\mathsf U}
\;\subseteq\;
E^{\mathrm{tr}}_{\mathsf U}
\;=\;
\eu\cap\etr.
\]
\end{corollary}

\begin{proof}
The canonical relation is endpoint-respecting, and bisimilar states have the same finite traces.
\end{proof}

Thus there are two canonical safe-forgetting partitions:
\begin{center}
\begin{tabular}{@{}lll@{}}
\toprule
Observation level & Canonical equivalence & Preserved structure \\
\midrule
Linear time & $\eu\cap\etr$ & finite rule-label traces \\
Branching time & $\approx_{\mathsf U}$ & exact labeled transition quotient \\
\bottomrule
\end{tabular}
\end{center}

The branching partition may be strictly finer because it preserves where choices occur, not only which words can eventually be realized.

\subsection{Greatest-fixed-point form}

For relations $R\subseteq S\times S$, define the monotone operator
\[
\Phi_{\mathsf U}(R)=\eu\cap\Bigl\{(s,t):
\begin{array}{l}
\text{for every }a\text{ and }s\step{a}s',\text{ some }t\step{a}t'\text{ has }(s',t')\in R,\\
\text{and conversely for every }a\text{ and }t\step{a}t'
\end{array}
\Bigr\}.
\]

\begin{proposition}[Coinductive characterization]
$\approx_{\mathsf U}$ is the greatest fixed point of $\Phi_{\mathsf U}$.
\end{proposition}

\begin{proof}
A relation $R$ is an endpoint-respecting bisimulation exactly when $R\subseteq\Phi_{\mathsf U}(R)$. By the Knaster--Tarski principle, the greatest post-fixed point of the monotone operator is its greatest fixed point. This greatest post-fixed point is precisely the union of all endpoint-respecting bisimulations.
\end{proof}

This form makes the criterion available for infinite systems even when no finite refinement algorithm terminates.

\section{Finite Partition Refinement}

Assume now that $S$ and the set of labels occurring on transitions are finite. Let $P_0$ be the partition of $S$ into endpoint-equivalence classes. For a partition $P$, define
\[
\Post_{P,a}(s)=\{B\in P: \exists t\in B,\ s\step{a}t\}
\]
and the signature
\[
\operatorname{sig}_P(s)
=
\left(
[s]_P,
\bigl(\Post_{P,a}(s)\bigr)_{a\in\Sigma}
\right).
\]
Let $F(P)$ be the partition whose blocks consist of states with equal $P$-signatures. Including $[s]_P$ in the signature ensures that $F(P)$ only splits blocks and never merges them.

Define iteratively
\[
P_{n+1}=F(P_n).
\]

\begin{theorem}[Termination and correctness of endpoint-initialized refinement]
For finite $S$, the sequence $(P_n)$ stabilizes after finitely many strict refinements. Its stable partition $P_\infty$ is exactly the partition into $\approx_{\mathsf U}$-classes.
\end{theorem}

\begin{proof}
Every strict refinement increases the number of blocks, which is bounded by $|S|$. Hence stabilization occurs after at most $|S|-|P_0|$ strict refinement rounds.

At stability, states in the same block have, for every label, identical sets of successor blocks. The equivalence relation induced by $P_\infty$ is therefore a strong bisimulation and refines endpoint equivalence.

For maximality, let $R$ be any endpoint-respecting bisimulation equivalence. We prove by induction that $R$ is contained in the equivalence induced by every $P_n$. The base case holds because $R\subseteq\eu$. Assume it holds for $P_n$, and let $sRt$. Then $s$ and $t$ lie in the same $P_n$-block. If $s\step{a}s'$, bisimulation supplies $t\step{a}t'$ with $s'Rt'$, so $s'$ and $t'$ lie in the same $P_n$-block. The converse direction is symmetric. Hence $s$ and $t$ have equal $P_n$-signatures and remain together in $P_{n+1}$. Thus every endpoint-respecting bisimulation is contained in the stable relation.

The stable relation is itself an endpoint-respecting bisimulation and contains every such relation, so it equals $\approx_{\mathsf U}$.
\end{proof}

This is a direct endpoint-initialized form of partition refinement, a standard technique for computing coarsest stable partitions and bisimulation quotients in finite transition systems \cite{paigetarjan1987}.

\begin{remark}[What the algorithm retains]
The algorithm does not retain provenance because it is provenance. It retains a provenance distinction only when that distinction changes the labeled successor-block structure, possibly after several rounds of propagation. A difference that never becomes behaviorally visible is quotiented away.
\end{remark}

\section{Application to Provenance-Decorated NRRG}

Consider the pointed typed graph system from the provenance-decorated NRRG framework \cite{carr2026nrrg}. Its visible endpoint projection forgets the provenance DAG and the inherited set of recombinant join-interface vertices. The relevant witness is unary, so the focal-parent convention for recombinant events plays no role in the separating transition.

The system contains two occurrences $\widetilde G$ and $\widetilde H$ whose visible pointed typed endpoints are both
\[
I\longrightarrow A\longrightarrow B.
\]
In $\widetilde G$, this endpoint was obtained without recombinant ancestry near the expandable vertex. Hence the recombinant ancestry radius is infinite and the boundary-touching rule $R_B$ is lawful. In $\widetilde H$, the visible interface is a recombinant join, the expandable vertex has radius $1$, and the same rule is unlawful. Figure~\ref{fig:nrrg-witness} displays the complete one-step distinction. The dashed halo is provenance decoration and is erased by $\mathsf U$; the visible endpoint graph itself is unchanged.

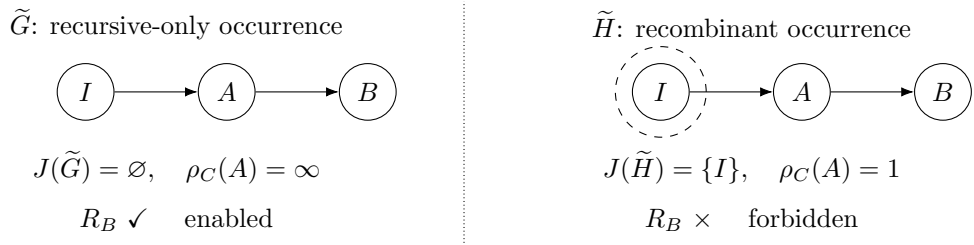
\begin{figure}[ht]
\centering
\begin{tikzpicture}[
  >=Latex,
  node distance=11mm,
  vnode/.style={draw,circle,minimum size=7.5mm,inner sep=0pt,font=\small},
  pmark/.style={draw,dashed,circle,minimum size=12mm,inner sep=0pt},
  lab/.style={font=\small,align=center}
]
\node[lab] (gt) at (-3.8,1.25) {$\widetilde G$: recursive-only occurrence};
\node[vnode] (gi) at (-5.0,0.35) {$I$};
\node[vnode,right=of gi] (ga) {$A$};
\node[vnode,right=of ga] (gb) {$B$};
\draw[->] (gi) -- (ga);
\draw[->] (ga) -- (gb);
\node[lab] (gr) at (-3.8,-0.65) {$J(\widetilde G)=\varnothing,\quad \rho_C(A)=\infty$};
\node[lab] at (-3.8,-1.35) {$R_B\ \checkmark\quad$ enabled};

\node[lab] (ht) at (3.8,1.25) {$\widetilde H$: recombinant occurrence};
\node[vnode] (hi) at (2.6,0.35) {$I$};
\node[pmark] at (hi) {};
\node[vnode,right=of hi] (ha) {$A$};
\node[vnode,right=of ha] (hb) {$B$};
\draw[->] (hi) -- (ha);
\draw[->] (ha) -- (hb);
\node[lab] (hr) at (3.8,-0.65) {$J(\widetilde H)=\{I\},\quad \rho_C(A)=1$};
\node[lab] at (3.8,-1.35) {$R_B\ \times\quad$ forbidden};

\draw[densely dotted] (0,-1.65) -- (0,1.55);
\end{tikzpicture}
\caption{Same visible endpoint, different enabled future. The dashed halo marks a recombinant join retained in provenance but forgotten by the endpoint projection.}
\label{fig:nrrg-witness}
\end{figure}

Consequently,
\[
\mathsf U(\widetilde G)\cong_X\mathsf U(\widetilde H),
\qquad
R_B\in\En(\widetilde G),
\qquad
R_B\notin\En(\widetilde H).
\]

\begin{proposition}[NRRG endpoint failure occurs at the enabled level]
The NRRG endpoint projection is not enabled-sufficient, and consequently is neither trace-sufficient nor quotient-sufficient.
\end{proposition}

\begin{proof}
The endpoint-equivalent occurrences in Figure~\ref{fig:nrrg-witness} have different enabled sets. The hierarchy theorem gives the remaining failures.
\end{proof}

\begin{corollary}[Immediate separation by refinement]
When partition refinement is initialized by visible NRRG endpoint classes, $\widetilde G$ and $\widetilde H$ are separated in the first refinement round.
\end{corollary}

\begin{proof}
Their signatures differ because their sets of successor blocks under $R_B$ differ: one has an $R_B$ successor and the other has none.
\end{proof}

The provenance theorem therefore supplies more than a counterexample to endpoint sufficiency. It identifies a pair that every exact endpoint-respecting behavioral abstraction is forced to distinguish.

\section{Interpretation}

\subsection{The corrected dependency statement}

The motivating statement can now be written without ambiguity:
\[
\boxed{
\begin{array}{c}
\text{endpoint trace sufficiency}
\end{array}
\iff
\eu\subseteq\etr
}
\]
and
\[
\boxed{
\begin{array}{c}
\text{endpoint quotient sufficiency}
\end{array}
\iff
\eu\text{ is a strong bisimulation equivalence.}
}
\]

Thus ``endpoint equivalence preserves lawful transition'' is the exact criterion for an exact branching quotient, not for every weaker notion of future agreement.

\subsection{Provenance as operational residue}

A provenance-decorated state may contain far more history than future behavior requires. The endpoint projection forgets all of it and can therefore forget too much. The canonical relations identify the intermediate target:

\begin{itemize}
    \item $E^{\mathrm{tr}}_{\mathsf U}$ retains exactly enough distinction to preserve finite lawful traces while respecting endpoints;
    \item $\approx_{\mathsf U}$ retains exactly enough distinction to preserve the full labeled branching quotient while respecting endpoints.
\end{itemize}

The difference between raw endpoint classes and these refined classes is the \emph{operational residue} of provenance: historical information that remains active in the future transition space.

\subsection{Abstraction and exactness}

The endpoint projection is an abstraction map. In abstract-interpretation language, the problem is not merely whether the abstraction is sound as an over-approximation, but whether it is exact for the chosen semantics \cite{cousot1977}. A raw quotient that unions transitions from all representatives is always constructible, but it can create behavior unavailable from a particular occurrence. Quotient sufficiency is precisely the condition preventing that representative-dependent inflation.

\subsection{Dependencies, logs, and lawful futures}

The formal consequence extends beyond graph ancestry. A dependency log matters whenever two artifacts with the same visible content have different lawful next operations because of their derivational histories. Examples may include software states, evidence chains, model lineages, legal instruments, scientific claims, or institutional procedures. The present theorems do not establish any such application automatically. They provide the test each application must pass:

\begin{quote}
Does endpoint equivalence lie inside the behavioral equivalence demanded by the application, and what is the coarsest endpoint-respecting refinement that restores it?
\end{quote}

\section{Scope and Further Work}

The current results concern finite traces and strong labeled transitions. Several extensions are immediate research directions.

\begin{enumerate}[label=(\arabic*)]
    \item \emph{Infinite behavior.} One may compare infinite traces, termination, divergence, fairness, or productivity. Different observations yield different safe-forgetting equivalences.
    \item \emph{Weak transitions.} Systems with silent administrative steps may require weak or branching bisimulation rather than strong bisimulation.
    \item \emph{True multi-parent semantics.} Recombinant generation can be represented by labeled hyperedges rather than focal-parent transitions. The correct quotient notion should then match entire parent interfaces, and the resulting sufficiency relation may differ from the focal-parent encoding studied here.
    \item \emph{Probabilistic and weighted generation.} If transitions carry probabilities, costs, or evidence weights, exact quotienting requires corresponding weighted bisimulation or lumpability conditions.
    \item \emph{Decidability under admissibility laws.} For infinite NRRG systems, the computability of $\approx_{\mathsf U}$ depends on how the admissibility law and provenance predicates are represented.
    \item \emph{Minimal provenance encodings.} The quotient partitions determine which occurrences must remain distinct, but not the most compact data structure encoding those distinctions.
\end{enumerate}

No universal claim about cultural, biological, linguistic, or computational evolution follows from the abstract results alone. Each domain must define its states, endpoint projection, labels, admissibility law, and required observation level.

\section{Conclusion}

A visible endpoint is not automatically a state. It is a projection, and the legitimacy of forgetting provenance depends on what future behavior the projection must preserve.

This paper separated three levels of sufficiency. Equality of enabled labels is weakest. Equality of finite trace languages is stronger and holds exactly when endpoint equivalence is contained in trace equivalence. Exact endpoint quotienting is stronger still and holds exactly when endpoint equivalence is a strong bisimulation equivalence. The implications are strict.

When endpoint equivalence fails, one need not choose between retaining the complete history and discarding it wholesale. The intersection $\eu\cap\etr$ gives the coarsest endpoint-respecting trace-safe refinement. The greatest endpoint-respecting bisimulation $\approx_{\mathsf U}$ gives the maximally coarse exact branching quotient, and every other exact endpoint-respecting quotient refines it. In finite systems, it can be computed by partition refinement initialized at the visible endpoint classes.

The central result is therefore a criterion of lawful forgetting:

\begin{quote}
Provenance may be erased exactly across those occurrence distinctions that do not alter the chosen lawful future semantics. What remains after maximal safe erasure is not historical decoration. It is generative state.
\end{quote}

\end{document}